\documentclass{scrartcl}

\usepackage[utf8]{inputenc}
\usepackage{amsmath,amsthm}
\usepackage[hidelinks,breaklinks]{hyperref}
\usepackage{breakurl}

\newtheorem{theorem}{Theorem}[section]
\newtheorem{definition}[theorem]{Definition}
\newtheorem{proposition}[theorem]{Proposition}

\usepackage[
    type={CC},
    modifier={by},
    version={4.0},
]{doclicense}

\title{Survey of Results on the ModPath and ModCycle Problems}
\author{Antoine Amarilli}
\date{}

\newcommand{\modpath}{\mathsf{ModPath}}
\newcommand{\modcycle}{\mathsf{ModCycle}}

\begin{document}
\maketitle

\begin{abstract}
  This note summarizes the state of what is known about the tractability of the
  problem ModPath, which asks if an input undirected graph contains a simple
  st-path whose length satisfies modulo
  constraints. We also consider the problem ModCycle, which asks for the
  existence of a simple cycle subject to such constraints. We also discuss the status of these problems on directed
  graphs, and on restricted classes of graphs. We explain connections to the problem variant asking
  for a constant vertex-disjoint number of such paths or cycles, and discuss
  links to other related work.
\end{abstract}

\section{Definition}

A \emph{simple path} in an undirected graph $G$ is a path that does not visit
the same vertex twice. (Note that we do not consider \emph{trails}, which can
reuse the same vertex twice but must not use the same edge twice.)

We study the following problem. It was 
posed in~\cite{martens2022complexity}, though related questions have been
studied much earlier (e.g., \cite{thomassen1983graph,arkin1991modularity}):

\begin{definition}
  Fix integers $p$ and $q>0$.
  Given an undirected graph $G$ and two vertices $s$ and $t$, we want to know if
  there exists a simple path connecting $s$ and $t$ in $G$ whose length is
  $p \bmod q$. We call this the $\modpath_{p,q}$ problem.
\end{definition}

A related problem is:

\begin{definition}
  Fix integers $p$ and $q>0$. Given an undirected graph $G$, we want to know if
  there is a simple cycle in $G$ whose length is $p \bmod q$. We call this the
  $\modcycle_{p,q}$ problem.
\end{definition}

The ModPath and ModCycle problems are clearly in NP, as we can easily check in polynomial time
that a path or cycle is suitable.
Note that, if we did not require the paths and cycles to be simple, then the
problems would be solvable in PTIME by a product construction: create $q$ copies
of the graph, create $q$ copies of each edge each of which is going from one
copy to the next, and solve a reachability problem involving vertices in the
zero-th and the $p$-th copy.

The question that we study is whether the ModPath and ModCycle problem are in
PTIME with the definition above which
requires simple paths and simple cycles. We study these problems and simple
variants of them in the rest of this note.

\section{Easy reductions}

\paragraph*{Between cycles and paths.}
It is easy to see that the $\modcycle$ problem reduces to the $\modpath$
problem:

\begin{proposition}
  \label{prp:cycle2path}
  For any integers $p$ and $q>0$, the problem $\modcycle_{p,q}$
  reduces in PTIME (with a Turing reduction) to the
  problem $\modpath_{p-1,q}$.
\end{proposition}

\begin{proof}
  We show how to reduce the problem, for each edge $e$ of the graph, of
  determining whether there is a cycle satisfying the length constraint and
  using the edge~$e$. This clearly suffices, as we can then simply test each
  possible choice of an edge~$e$. For the reduction,
  we modify the graph $G$ to a graph $G_e$ where the edge $e$ is
  removed, and where the source and sink to the $\modpath$ problem 
  are the endpoints of the edge~$e$. Clearly
  there is a bijection between cycles in~$G$ using the edge~$e$ and
  $st$-paths in~$G_e$, and this bijection maps cycles of length $p \bmod q$ to
  paths of length $p-1 \bmod q$.
\end{proof}

We can show a similar reduction in the other direction, but the modulo is
changed:

\begin{proposition}
  \label{prp:path2cycle}
  For any integers $p$ and $q>0$, the problem $\modpath_{p,q}$ reduces in PTIME
  (with a Karp reduction) to the problem $\modcycle_{2p+1, 2q}$.
\end{proposition}

\begin{proof}
  Given the input undirected graph $G$ to $\modpath_{p,q}$ with source and sink
  $s$ and $t$, subdivide each edge twice, and add a single edge connecting $s$
  and $t$. We let $G'$ be the result. Given a path of length $p \bmod q$
  connecting $s$ and $t$ in~$G$, we deduce a path of length $2p \bmod 2q$
  connecting $s$ and $t$ in~$G'$, hence a cycle of length $2p+1 \bmod 2q$ in~$G'$
  thanks to the extra edge. Conversely, we see that all cycles in $G'$ have
  even length in~$G$ except if they use the additional edge between $s$ and $t$,
  in which case they give us a simple path between $s$ and $t$. If the cycle has
  length $2p+1 \bmod 2q$ with the extra edge, then the simple path in question
  has length $2p \bmod 2q$ in~$G'$, hence length $p \bmod q$ in~$G$.
\end{proof}

We are not aware of a reduction from $\modpath$ to $\modcycle$ which preserves
the modulo.

\paragraph*{On remainders.}
It is also clear that, for the $\modpath$ problem, the value of the reminder
does not matter:

\begin{proposition}
  \label{prp:remchange}
  For any integers $p$ and $p'$ and $q>0$, the problem $\modpath_{p,q}$
  reduces in PTIME (with a Karp reduction) to the
  problem $\modpath_{p',q}$.
\end{proposition}

\begin{proof}
  Simply add a path of the suitable length connecting $t$ to a new vertex~$t'$,
  and reduce to $\modpath_{p',q}$ with source $s$ and target $t'$.
\end{proof}

Our results in Section~\ref{sec:cyclezero} will imply that the same is not true
of $\modcycle$ (assuming that P is different from NP).

\paragraph*{On moduli.}
It is also obvious that, for paths, the problem is at least as hard when taking
a multiple of the original modulo.

\begin{proposition}
  For any integers $p$ and $q>0$ and $k > 0$, the problem $\modpath_{p,q}$
  reduces in PTIME (with a Turing reduction) to the
  problem $\modpath_{p,kq}$.
\end{proposition}

\begin{proof}
  There is a path connecting $s$ and $t$ with length $p \bmod q$ iff there is a
  path of length $p+k'q \bmod kq$ for some~$k$, so we can conclude using the
  oracle and using Proposition~\ref{prp:remchange}.
\end{proof}

A similar reduction works from $\modcycle_{p,q}$ if we assume an oracle for the problems
$\modcycle_{p',kq}$  with $p' = p + iq$ for every $i$.

\section{On directed graphs}

In this section, we discuss the status of the problems $\modpath$ and
$\modcycle$ when studying them on directed graphs instead of undirected graphs.
We exclude the trivial case of the modulo $q=1$  as the problems then amount to
reachability or to testing the existence of a directed cycle which are clearly
solvable in polynomial time.

\paragraph*{Paths.}
The analogue of the $\modpath$ problem on directed graphs
is NP-hard.
Indeed, the following is known (with an elementary but non-trivial
proof):

\begin{proposition}[\cite{fortune1980directed}]
\label{prp:disjoint}
  The problem, given a directed graph $G$ and vertices $s, t, s', t'$, of
  deciding if there is a path from $s$ to $t$ and a path from $s'$ to $t'$ that
  are vertex-disjoint, is NP-hard.
\end{proposition}

Proposition~\ref{prp:disjoint} implies 
that $\modpath$ is hard on directed graphs:

\begin{proposition}
\label{prp:directed}
  Fix any $p$ and $q \geq 2$. The problem, given a directed graph $G$ and
  vertices $s$ and $t$, of testing if there is a simple path of length $p$ mod
  $q$ from $s$ to $t$, is NP-hard.
\end{proposition}

\begin{proof}
  We reduce from Proposition~\ref{prp:disjoint}.
  We first assume $p>0$. Then, given a directed graph $G$ with vertices
  $s,t,s',t'$, we replace each edge by a path of $q$ edges, choose $s$ as the
  source and $t'$ as the sink, and add a path of $p$ edges from $t$ to~$s'$.
  Then it is clear that any path of length $p \bmod q$ from~$s$ to~$t$ must use
  the path of $p$ edges, and thus give a solution to the problem of
  Proposition~\ref{prp:disjoint}. Conversely, any solution to the latter problem
  gives a path for the former problem.

  If $p=0$, we subdivide $G$ as indicated, we choose as source a fresh vertex
  with a 1-edge path to~$s$, choose $t'$ as sink, and add a path of $q-1$ edges
  from~$t$ to~$s'$. The reasoning is similar.
\end{proof}

\paragraph*{Cycles.}
Intractability also holds for the problem $\modcycle$ on directed graphs in the
case where $p\neq 0$ and $q$ is large
enough, as was already observed in~\cite{arkin1991modularity}:

\begin{proposition}
  Fix any $0 < p < q$ and $q \geq 3$. The problem, given a directed graph $G$, of
  testing if there is a simple cycle of length $p \bmod q$, is NP-hard.
\end{proposition}

\begin{proof}
  We first show that there are two values $0 < p_1, p_2 < q$  such that $p
  \neq p_1$, $p \neq p_2$, and $p = p_1
  + p_2 \bmod q$. If $p \geq 2$ then this is clear, taking $p_1 = 1$ and $p_2 =
  p_1 - 1$. If $p = 1$, take $p_1 = 2$ and $p_2 = q-1$, noting that $p_1 < q$
  and $p_2 \neq q$ because $q \geq 3$.

  We reduce from Proposition~\ref{prp:disjoint} like in the first case of Proposition~\ref{prp:directed}: given the
  directed graph $G$, we subdivide each edge to a path of length $q$, then add a
  directed path of length $p_1$ from $t$ to $s'$ and a directed path of length
  $p_2$ from $t'$ to~$s$. Now, from a solution to the problem of
  Proposition~\ref{prp:disjoint}, we deduce a cycle of length $p_1 + p_2 = p \bmod q$.
  Conversely, in~$G$, we can partition the cycles among those that use either
  none of the extra paths, one of the extra paths, or both extra paths. Their
  modulo values are $0$, $p_1$ or $p_2$, and $p$ respectively. Now, as $p_1 \neq
  p$ and $p_2 \neq p$ and $p \neq 0$, this means that a cycle of length $p \bmod q$
  must use both extra paths. This gives us two disjoint paths from $s$ to~$t$
  and from $s'$ to~$t'$ in~$G$, hence in the initial graph.
\end{proof}

A more general complexity classification is given
in~\cite{hemaspaandra2004complexity} for a variant of
the problem where the remainder modulo the value~$q$ is required to fall in a
certain set $S$ of allowed remainders (instead of $S=\{p\}$), provided that the set $S$ of allowed
remainders does not
include~$0$.

For the case where the requested remainder is $p=0$, i.e., the problem
$\modcycle_{0,q}$ with $q \geq 3$ on directed graphs, the complexity appears to
be open: this is stated as open in~\cite{hemaspaandra2004complexity}. In
other words, for any fixed $q \geq 3$, it is open whether we can determine
in PTIME, given a directed graph, whether it contains a simple cycle of length
multiple of~$q$.
Also note that
this same problem with $p=0$ is known to tractable on undirected graphs
(Section~\ref{sec:cyclezero}).

For the case $q=2$, it is known that we can (easily) test in PTIME whether a directed graph
contains an odd cycle~\cite{thomassen1985even}, and (less easily) whether it contains an even
cycle~\cite{robertson1999permanents,mccuaig2004polyas}. Accordingly:

\begin{proposition}[\cite{thomassen1985even,robertson1999permanents,mccuaig2004polyas}]
  \label{prp:cycledirected2}
  For $p\in \{0,1\}$, the problem, given a directed graph $G$, of
  testing if there is a simple cycle of length $p \bmod 2$, is in PTIME.
\end{proposition}

Note that this contrasts with the intractability of the same task for $\modpath$
with $q=2$ (Proposition~\ref{prp:directed}).

Incidentally, only
very recently was a tractable randomized algorithm shown to compute the
\emph{shortest}
simple even cycle in a directed graph~\cite{bjorklund2022shortest}.
The problem of the shortest simple odd cycle in a directed graph
can easily be seen to be in polynomial time, thanks to the fact
that the shortest odd cycle is necessarily simple~\cite{cstheory}.

\paragraph*{Restricted classes of directed graphs.}
One immediate observation is that all problems on directed graphs discussed in this section are
tractable if the input is assumed to be a \emph{directed acyclic graph}: such
graphs have no cycles, and the
directed paths on such graphs are automatically simple.

\section{Testing if all moduli are the same, and the case of modulo 2}
\label{sec:allsame}

\paragraph*{Paths.}
It is shown in~\cite[Theorem~4]{arkin1991modularity} that, for any fixed
$p$ and~$q$,
one can test in PTIME, given a
graph and vertices $s$ and $t$, whether \emph{all simple paths} connecting $s$
and $t$ in~$G$ have length $p \bmod q$.
This implies in particular the following:

\begin{proposition}
  The problems $\modpath_{0,2}$ and $\modpath_{1,2}$ are in PTIME.
\end{proposition}

The same result can be shown with a simpler proof due to Edmonds,
see~\cite[Section II]{lapaugh1984even} for the case $p=0$ and $q=2$, which
implies the case of $p=1$ via Proposition~\ref{prp:remchange}.

\paragraph*{Cycles.}
The same tractability result holds for the $\modcycle$ problem. In fact, tractability
even holds on
directed graphs as we have seen (Proposition~\ref{prp:cycledirected2}); but it can be shown
to hold with
an easy proof in the case of undirected graphs:

\begin{proposition}[mentioned in \cite{thomassen1985even}]
  The problems $\modcycle_{0,2}$ and $\modcycle_{1,2}$ are in PTIME.
\end{proposition}

\begin{proof}
For undirected graphs and modulo two, an undirected graph has an odd cycle
unless it is bipartite (which can be checked in PTIME), and it has an even cycle
unless every biconnected component is an odd cycle or a single edge (which can
be checked in PTIME).
\end{proof}

\section{Cycles when the remainder is zero}
\label{sec:cyclezero}

For the problem $\modcycle$, when the requested remainder is 0,
then the problem is known to be tractable:

\begin{proposition}[\cite{thomassen1988presence}]
  \label{prp:cyclezero}
  For any $q>0$, the problem $\modcycle_{0,q}$ is in PTIME.
\end{proposition}

This is because any large-treewidth
graph must contain such a cycle. Specifically, \cite[Proposition
3.2]{thomassen1988presence} shows that any high-treewidth graph contains as
topological minor a wall graph where all edges are $0 \bmod q$.
Thus, the answer is
yes on high-treewidth graphs, and on low-treewidth graphs the problems are
always tractable (see Section~\ref{sec:treewidth}).


\section{Bounded-treewidth graphs}
\label{sec:treewidth}

Under the assumption that the graphs have bounded treewidth, then the problem
$\modpath$ (hence, $\modcycle$ by Proposition~\ref{prp:cycle2path}) is always in
PTIME:

\begin{proposition}[Theorem~5.2, \cite{thomassen1988presence}]
  Let $p$ and $q$ and $k$ be arbitrary integers. The problem $\modpath_{p,q}$ is
  in PTIME if we assume that the input graphs have treewidth at most~$k$.
\end{proposition}

\begin{proof}
  One intuition is that we can process the graph along a tree decomposition
  and solve the problem by dynamic programming. Intuitively, we can
  remember, for each bag, for each set of disjoint pairs of endpoints and modulo
  lengths, which such sets are achievable simultaneously by disjoint paths in
  the subgraph induced by the nodes occurring below that bag in the tree
  decomposition.
\end{proof}

\section{Constraints on degree and connectivity}

In this section, we review some combinatorial results which study which
conditions on the graph can guarantee the existence of cycles whose lengths
achieve some prescribed remainder values. The conditions studied are minimal
degree, average degree (or, equivalently, edge density), and connectivity.

\paragraph*{Minimal degree.}
There are results showing that, when graphs are asserted to have sufficiently
high minimal degree, then it is impossible to avoid some cycle lengths.
Specifically, for graphs of
sufficiently large ($O(q)$) minimal degree, then there must by cycles of all
even lengths modulo~$q$, and if the graph is 2-connected and not bipartite then
there must be cycles of all lengths modulo~$q$~\cite{gao2022unified}.
We note that the main result of~\cite{gao2022unified} is actually a result about
path lengths on graphs with sufficiently high minimal degree, so this also gives
results on the $\modpath$ problem in this setting.

\paragraph*{Average degree.}
There are also existence results based on
edge density~\cite{bollobas1977cycles} (or equivalently on average degree).
Specifically, for any odd modulo $q$,
considering graphs with a sufficiently high number of edges (i.e., at least $c_q n$ 
some constant $q$, where $n$ is the number of vertices), then such graphs must contain
a cycle of length $p$ mod $q$ for every~$p$. (Of course the same cannot be true
for even~$q$, e.g., considering bipartite graphs.)
A similar result holds for
all even remainders, i.e., cycles of length $2p$ modulo $q$ for
arbitrary~$q$~\cite{verstraete2000arithmetic}.

\paragraph*{Connectivity.}
There are also existence results based on connectivity. 
It is known that, for all $q \geq 3$, every $q$-connected graph contains a cycle
of length zero modulo~$q$~\cite{gao2022unified}; note that the same is true
assuming that the treewidth is sufficiently high (Section~\ref{sec:cyclezero}).
Every $q$-connected graph must also contain cycles of all even lengths
modulo~$q$ provided that $q \geq 6$~\cite[Theorem~5.16]{gao2022unified}.
Other
results are known about the existence of $k$-linkages with modulo conditions
assuming sufficiently high connectivity~\cite{chen2009modulo}.
Last, it is shown in \cite{lyngsie2021cycle} that for \emph{odd} moduli $q$, any
sufficiently large 3-connected cubic graph contains cycles with each possible
length modulo $q$.

\paragraph*{Directed graphs.}
Some combinatorial results are also known for directed graphs.
It is known that strongly connected directed graphs of sufficiently high edge
density must contain an even cycle~\cite{chung1994even}.
Lower bounds on degree also imply the
existence of cycles with remainder 0 in the case of directed
graphs~\cite{alon1989cycles}, and there are similar bounds on the dichromatic
number~\cite{steiner2022subdivisions}.

\section{Multiple paths or multiple cycles}

One natural generalization of $\modpath$ and $\modcycle$ is to ask for the
existence of $k$ disjoint paths or $k$ disjoint cycles satisfying the
conditions. These problems have been studied when $k$ is given as input: in this
case
the problem is NP-hard for paths (as a special
case of discrete multicommodity flow~\cite{karp1975computational}) and
for cycles (already if we want to partition a graph on $3n$ vertices into $n$
vertex-disjoint triangles; see
\cite{kirkpatrick1978completeness} or \cite{vanrooij2013partition}). They have
also been studied when
parameterizing by~$k$, e.g., \cite{bodlaender2013kernel}.
Here we assume that $k$ is a constant.

\paragraph*{Without modulo constraints.}
The results of
Robertson on Seymour imply that, on \emph{undirected graphs}, for any constant~$k$,
given source-sink pairs $(s_1, t_1), \ldots, (s_k, t_k)$,  we
can decide in PTIME whether there exist $k$ pairwise disjoint paths each of
which connects $s_i$ and $t_i$~\cite{robertson1995graph13}.
(Note that there are recent extensions
of such results to problems where we want to minimize the \emph{total length} of
such paths, in the case $k=2$~\cite{bjorklund2019shortest}.) 
It is also known that we can test in linear time whether an input undirected
graph contains $k$ vertex-disjoint cycles~\cite{bodlaender1994disjoint}, note
that this easily follows from a treewidth-based argument.

Now, on directed graphs, 
the existence of $k$ vertex-disjoint cycles can also be
tested in PTIME even on directed graphs \cite[Section 5]{reed1996packing}. By
contrast, asking for
the existence of two disjoint paths for two source-target pairs is NP-hard (see
Proposition~\ref{prp:disjoint}). Note that this last result no longer holds if the graph is required to be
\emph{planar}: in this case, the disjoint path problems for any constant $k$ can be
solved in polynomial time~\cite{schrijver1994finding}.

Also note that these results have been extended in the setting where we are
looking for $k$-tuples of paths that must be \emph{shortest paths} from~$s_i$
to~$t_i$. This task is known to be tractable on undirected
graphs~\cite{lochet2021polynomial}, and on directed planar graphs or on directed
graphs with $k=2$~\cite{berczi2017directed}.

\paragraph*{With modulo constraints.}
The results on $k$ disjoint paths and cycles on undirected graphs have been extended
to test the existence of some
constant number of disjoint cycles or paths of prescribed modulo values. It is known that
on undirected graphs you can test in PTIME for the
presence of a constant number of disjoint cycles of
length divisible by some~$q$~\cite[Theorem~5.1]{thomassen1988presence}; note
that this follows from the proof of
Proposition~\ref{prp:cyclezero}, as the answer is always yes on graphs of
sufficiently high treewidth.
It was shown in~\cite{kawarabayashi2010odd} that you can test in
PTIME for the existence of $k$ vertex-disjoint odd cycles in undirected graphs.
You can also test in PTIME for the existence of $k$ vertex-disjoint
paths connecting $k$ source-sink pairs with prescribed parities~\cite{kawarabayashi2011graph}.

Very recently~\cite{kawarabayashi2023half}, a tractability result for cycles
with parity constraints was shown on directed graphs: you can test in PTIME
on an input directed graph whether it contains 
$k$ vertex-disjoint odd cycles.

\section{Other related work}

\paragraph*{Group-labeled graphs.}
Another model is that of group-labeled graphs, for which there is a 
\emph{directed} and \emph{undirected} setting. 
The setting of \emph{directed group-labeled graphs}~\cite{kawase2020finding},
considers \emph{directed} graphs
where edges are labeled by an element of an abelian group: the value of a
directed path is the composition of the labels of the edges traversed by the
path. However, edges can then also be traversed in a reverse direction, in which
case we compose by the inverse of their label.

Hence, except in cases like $q =
2$, the model of directed group-labeled graphs does not seem well-adapted to code the $\modpath$ or $\modcycle$
problems.
Note that the undirected and directed settings are equivalent where all nonzero elements
of the group have order two~\cite{thomas2023packing}.

The setting of \emph{undirected group-labeled graphs}~\cite{wollan2010packing,wollan2011packing} is
closer to our problem. In this setting, the model considers \emph{undirected}
graphs with each edge again labeled by an element of an abelian group, and with
the value of a path being the combination of the edge labels.

In the setting of undirected group-labeled graphs, it
was recently shown~\cite{thomas2023packing} that, for any prime power~$q$, the
cycles of length $p \bmod q$ satisfy the Erdős-Pósa property for all $p$.
There are other similar results in this model on the Erdős-Pósa property for paths
with prescribed endpoints and modulo values.
There are other earlier results on
cycles~\cite{gollin2021unified,gollin2022unified}.
However, this does not seem to imply any result on the complexity of detecting
whether such cycles or paths are present in an input graph. 

\paragraph*{Robertson-Seymour with parity conditions.}
There is work aiming at generalizing Robertson-Seymour results with parity
conditions~\cite{kawarabayashi2011graph,kawarabayashi2013totally}.
However, these do not seem to have been extended to moduli greater than~$2$,
and our problems are known to be tractable for $q=2$
(Section~\ref{sec:allsame}).

\paragraph*{Graphs with large clique minors.}
It is known that graphs with large clique minors must contain certain subgraphs
with edges interpreted as multiples of some
value~\cite{alon2021divisible,das2021tight}.

\paragraph*{Expanding graphs.}
It is known that expanders, aka expanding
graphs, must contain cycles of all moduli \cite{martinsson2023cycle}.

\paragraph*{Planar graphs.}
It is known that on cubic, 3-connected, planar graphs, between any two
vertices of an undirected graph there must be
paths of all remainders modulo $q=3$, and such paths can be found in polynomial
time~\cite{deng1991path}.

\bibliographystyle{alpha}
\bibliography{main}

\newcommand{\etalchar}[1]{$^{#1}$}
\begin{thebibliography}{vRvKNB13}

\bibitem[AK21]{alon2021divisible}
Noga Alon and Michael Krivelevich.
\newblock Divisible subdivisions.
\newblock {\em Journal of Graph Theory}, 98(4), 2021.

\bibitem[AL89]{alon1989cycles}
Noga Alon and Nathan Linial.
\newblock Cycles of length 0 modulo $k$ in directed graphs.
\newblock {\em Journal of Combinatorial Theory, Series B}, 47(1), 1989.

\bibitem[APY91]{arkin1991modularity}
Esther~M Arkin, Christos~H Papadimitriou, and Mihalis Yannakakis.
\newblock Modularity of cycles and paths in graphs.
\newblock {\em JACM}, 38(2), 1991.

\bibitem[BH19]{bjorklund2019shortest}
Andreas Bj\"orklund and Thore Husfeldt.
\newblock Shortest two disjoint paths in polynomial time.
\newblock {\em SIAM Journal on Computing}, 48(6), 2019.

\bibitem[BHK22]{bjorklund2022shortest}
Andreas Bj{\"o}rklund, Thore Husfeldt, and Petteri Kaski.
\newblock The shortest even cycle problem is tractable.
\newblock In {\em STOC}, 2022.

\bibitem[BJK13]{bodlaender2013kernel}
Hans~L Bodlaender, Bart~MP Jansen, and Stefan Kratsch.
\newblock Kernel bounds for path and cycle problems.
\newblock {\em Theoretical Computer Science}, 511, 2013.

\bibitem[BK17]{berczi2017directed}
Krist{\'o}f B{\'e}rczi and Yusuke Kobayashi.
\newblock The directed disjoint shortest paths problem.
\newblock In {\em ESA}, 2017.

\bibitem[Bod94]{bodlaender1994disjoint}
Hans~L Bodlaender.
\newblock On disjoint cycles.
\newblock {\em International Journal of Foundations of Computer Science},
  5(01), 1994.

\bibitem[Bol77]{bollobas1977cycles}
Bela Bollob{\'a}s.
\newblock Cycles modulo $k$.
\newblock {\em Bulletin of the London Mathematical Society}, 9(1), 1977.

\bibitem[Cad22]{cstheory}
Caduk:\url{https://math.stackexchange.com/users/966819/caduk}.
\newblock Algorithm for finding shortest directed odd cycle in a digraph.
\newblock Mathematics Stack Exchange, 2022.
\newblock URL:\url{https://math.stackexchange.com/q/4436683} (version:
  2022-04-26).

\bibitem[CGK94]{chung1994even}
Fan~RK Chung, Wayne Goddard, and Daniel~J Kleitman.
\newblock Even cycles in directed graphs.
\newblock {\em SIAM Journal on Discrete Mathematics}, 7(3), 1994.

\bibitem[CMZ09]{chen2009modulo}
Yuan Chen, Yao Mao, and Qunjiao Zhang.
\newblock On modulo linked graphs.
\newblock In {\em FAW}, 2009.

\bibitem[DDS21]{das2021tight}
Shagnik Das, Nemanja Dragani{\'c}, and Raphael Steiner.
\newblock Tight bounds for divisible subdivisions.
\newblock {\em arXiv preprint arXiv:2111.05723}, 2021.

\bibitem[DP91]{deng1991path}
Xiaotie Deng and Christos~H Papadimitriou.
\newblock On path lengths modulo three.
\newblock {\em Journal of graph theory}, 15(3), 1991.

\bibitem[FHW80]{fortune1980directed}
Steven Fortune, John Hopcroft, and James Wyllie.
\newblock The directed subgraph homeomorphism problem.
\newblock {\em Theoretical Computer Science}, 10(2), 1980.

\bibitem[GHK{\etalchar{+}}21]{gollin2021unified}
J~Pascal Gollin, Kevin Hendrey, Ken-ichi Kawarabayashi, O~Kwon, Sang-il Oum,
  et~al.
\newblock A unified half-integral {E}rd{\H{o}}s-{P}{\'{o}}sa theorem for cycles
  in graphs labelled by multiple abelian groups.
\newblock {\em arXiv preprint arXiv:2102.01986}, 2021.

\bibitem[GHK{\etalchar{+}}22]{gollin2022unified}
J~Pascal Gollin, Kevin Hendrey, O-j Kwon, Sang-il Oum, Youngho Yoo, et~al.
\newblock A unified {E}rd{\H{o}}s-{P}{\'{o}}sa theorem for cycles in graphs
  labelled by multiple abelian groups.
\newblock {\em arXiv preprint arXiv:2209.09488}, 2022.

\bibitem[GHLM22]{gao2022unified}
Jun Gao, Qingyi Huo, Chun-Hung Liu, and Jie Ma.
\newblock A unified proof of conjectures on cycle lengths in graphs.
\newblock {\em International Mathematics Research Notices}, 2022(10), 2022.

\bibitem[HST04]{hemaspaandra2004complexity}
Edith Hemaspaandra, Holger Spakowski, and Mayur Thakur.
\newblock Complexity of cycle length modularity problems in graphs.
\newblock In {\em LATIN}, 2004.

\bibitem[Kar75]{karp1975computational}
Richard~M Karp.
\newblock On the computational complexity of combinatorial problems.
\newblock {\em Networks}, 5(1), 1975.

\bibitem[Kaw13]{kawarabayashi2013totally}
Ken-ichi Kawarabayashi.
\newblock Totally odd subdivisions and parity subdivisions: Structures and
  coloring.
\newblock In {\em SODA}, 2013.

\bibitem[KH78]{kirkpatrick1978completeness}
David~G Kirkpatrick and Pavol Hell.
\newblock On the completeness of a generalized matching problem.
\newblock In {\em STOC}, 1978.

\bibitem[KKKX23]{kawarabayashi2023half}
Ken-ichi Kawarabayashi, Stephan Kreutzer, O-joung Kwon, and Qiqin Xie.
\newblock A half-integral erd{\H{o}}s-p{\'o}sa theorem for directed odd cycles.
\newblock In {\em SODA}, 2023.

\bibitem[KKY20]{kawase2020finding}
Yasushi Kawase, Yusuke Kobayashi, and Yutaro Yamaguchi.
\newblock Finding a path with two labels forbidden in group-labeled graphs.
\newblock {\em Journal of Combinatorial Theory, Series B}, 143, 2020.

\bibitem[KR10]{kawarabayashi2010odd}
Ken-ichi Kawarabayashi and Bruce Reed.
\newblock Odd cycle packing.
\newblock In {\em STOC}, 2010.

\bibitem[KRW11]{kawarabayashi2011graph}
Ken-ichi Kawarabayashi, Bruce Reed, and Paul Wollan.
\newblock The graph minor algorithm with parity conditions.
\newblock In {\em FOCS}, 2011.

\bibitem[LM21]{lyngsie2021cycle}
Kasper~Szabo Lyngsie and Martin Merker.
\newblock Cycle lengths modulo $k$ in large 3-connected cubic graphs.
\newblock {\em Advances in Combinatorics}, 2021.

\bibitem[Loc21]{lochet2021polynomial}
Willian Lochet.
\newblock A polynomial time algorithm for the k-disjoint shortest paths
  problem.
\newblock In {\em SODA}, 2021.

\bibitem[LP84]{lapaugh1984even}
Andrea~S LaPaugh and Christos~H Papadimitriou.
\newblock The even-path problem for graphs and digraphs.
\newblock {\em Networks}, 14(4), 1984.

\bibitem[McC04]{mccuaig2004polyas}
William McCuaig.
\newblock P{\'o}lya's permanent problem.
\newblock {\em The Electronic Journal of Combinatorics}, 2004.

\bibitem[MP22]{martens2022complexity}
Wim Martens and Tina Popp.
\newblock The complexity of regular trail and simple path queries on undirected
  graphs.
\newblock In {\em PODS}, 2022.

\bibitem[MS23]{martinsson2023cycle}
Anders Martinsson and Raphael Steiner.
\newblock Cycle lengths modulo $k$ in expanders.
\newblock {\em European Journal of Combinatorics}, 109, 2023.

\bibitem[RRST96]{reed1996packing}
Bruce Reed, Neil Robertson, Paul Seymour, and Robin Thomas.
\newblock Packing directed circuits.
\newblock {\em Combinatorica}, 16(4), 1996.

\bibitem[RS95]{robertson1995graph13}
Neil Robertson and Paul~D Seymour.
\newblock Graph minors. {XIII}. {T}he disjoint paths problem.
\newblock {\em Journal of combinatorial theory, Series B}, 63(1), 1995.

\bibitem[RST99]{robertson1999permanents}
Neil Robertson, Paul~D Seymour, and Robin Thomas.
\newblock Permanents, pfaffian orientations, and even directed circuits.
\newblock {\em Annals of mathematics}, 1999.

\bibitem[Sch94]{schrijver1994finding}
Alexander Schrijver.
\newblock Finding $k$ disjoint paths in a directed planar graph.
\newblock {\em SIAM Journal on Computing}, 23(4), 1994.

\bibitem[Ste22]{steiner2022subdivisions}
Raphael Steiner.
\newblock Subdivisions with congruence constraints in digraphs of large
  chromatic number.
\newblock {\em arXiv preprint arXiv:2208.06358}, 2022.

\bibitem[Tho83]{thomassen1983graph}
Carsten Thomassen.
\newblock Graph decomposition with applications to subdivisions and path
  systems modulo $k$.
\newblock {\em Journal of Graph Theory}, 7(2), 1983.

\bibitem[Tho85]{thomassen1985even}
Carsten Thomassen.
\newblock Even cycles in directed graphs.
\newblock {\em European Journal of Combinatorics}, 6(1), 1985.

\bibitem[Tho88]{thomassen1988presence}
Carsten Thomassen.
\newblock On the presence of disjoint subgraphs of a specified type.
\newblock {\em Journal of Graph Theory}, 12(1), 1988.

\bibitem[TY23]{thomas2023packing}
Robin Thomas and Youngho Yoo.
\newblock Packing {A}-paths of length zero modulo a prime.
\newblock {\em Journal of Combinatorial Theory, Series B}, 160, 2023.

\bibitem[Ver00]{verstraete2000arithmetic}
Jacques Verstra{\"e}te.
\newblock On arithmetic progressions of cycle lengths in graphs.
\newblock {\em Combinatorics, Probability and Computing}, 9(4), 2000.

\bibitem[vRvKNB13]{vanrooij2013partition}
Johan~MM van Rooij, Marcel~E van Kooten~Niekerk, and Hans~L Bodlaender.
\newblock Partition into triangles on bounded degree graphs.
\newblock {\em Theory of Computing Systems}, 52(4), 2013.

\bibitem[Wol10]{wollan2010packing}
Paul Wollan.
\newblock Packing non-zero {A}-paths in an undirected model of group labeled
  graphs.
\newblock {\em Journal of Combinatorial Theory, Series B}, 100(2), 2010.

\bibitem[Wol11]{wollan2011packing}
Paul Wollan.
\newblock Packing cycles with modularity constraints.
\newblock {\em Combinatorica}, 31(1), 2011.

\end{thebibliography}

\vfill
\doclicenseThis

\end{document}